\documentclass{llncs}
\usepackage{amsmath}
\usepackage{amssymb}
\usepackage{amsfonts}
\usepackage{xcolor}
\usepackage[english]{babel}
\usepackage[utf8x]{inputenc}
\usepackage[T1]{fontenc}
\usepackage{times}
\usepackage{tikz}
\usepackage{subfig}
\usepackage{algorithm2e}
\usepackage{listings}
\usepackage{multicol}
\usepackage{graphicx}
\usepackage{url}
\usepackage{hyperref}
\usepackage[colorinlistoftodos]{todonotes}

\usetikzlibrary{arrows,automata}
\begin{document}

\title{Fault Detection for Timed FSM with Timeouts by  Constraint Solving}

\author{ Omer Nguena Timo \and Dimitri Prestat \and  Florent Avellaneda}

\institute{Computer Research Institute of Montreal, CRIM\\
Montreal, Canada \\   \email{\{omer.nguena-timo, dimitri.prestat,  florent.avellaneda\}@crim.ca}
}
\maketitle

\begin{abstract}
Recently, an efficient constraint solving-based approach has been developed  to detect logical faults in systems specified with classical finite state machines (FSMs). The approach is unsuitable to detect  violations of time constraints. 
In this paper, we lift the approach to generated tests detecting both logical faults and violations of time constraints in  systems specified with timed FSMs with timeouts (TFSMs-T). We propose a method to verify whether a given test suite is complete, i.e., it detects all the faulty implementations in a fault-domain and a method to generate a complete test suite.  We conduct experiments to evaluate the scalability of the proposed methods.
\end{abstract}


\section{Introduction} 
The fault domain coverage criterion can be adopted to generate  tests   revealing  faults in safety/security critical systems under test (SUT)~\cite{Androutsopoulos:2013}.
The domain can be built from referenced databases\footnote{E.g.: \url{https://nvd.nist.gov/}} or  expert knowledge. Efficient test generation methods are needed especially for the fault domains of important sizes, which has motivated the development of  an approach~\cite{Petrenko2016a,Petrenko2016b} leveraging on recent advances in the field of (Boolean) constraint solving. 
The approach has been elaborated  to  detect logical faults in reactive systems specified with finite state machines (FSMs).  
We plan to lift the approach to detect both logical faults and violations of time constraints in reactive systems; especially, we focus on reactive systems specified with timed FSMs with timeouts (TFSMs-T).

\par TFSM-T~\cite{Merayo2008a,2014arXiv1408.5967B} is 
an extension of the classical FSM with timeout transitions for expressing time constraints. They define  timeouts and the next states to be reached if no input is applied before the timeouts expire; otherwise outputs defined by  input/output transitions are produced. 
Although they express limited types of time constraints as compared to other timed FSMs~\cite{2014arXiv1408.5967B}, TFSMs-T have been used to specify reactive systems such as web applications~\cite{ZhigulinYMC11} and protocols~\cite{Zhigulin2012,RFC1350,Huang1998}. 
Logical faults in TFSMs-T
correspond to  unexpected outputs or unexpected state changes.
Reducing and increasing waiting time are violations of time constraints. An implementation under test for a given specification TFSM-T can be represented with a mutated version of the specification TFSM-T also called a mutant. A mutant can be obtained  by seeding the specification with an arbitrary number of faults.  A fault domain for a specification is then a finite set of possible mutants; it can be built from a list of identified faults to be detected  in systems under test. A mutant is nonconforming if its timed output sequence differs from that of the specification  for some  timed input sequences (tests). A complete test suite for a fault domain detects all nonconforming mutants in the domain. 
\par Model-based testing with guaranteed fault coverage has been investigated for untimed and timed models. 
Finite state machines  can be preferred over label transition systems  for representing systems. This is probably because  FSMs have been  used early in testing digital circuits~\cite{YANNAKAKIS1995} and protocols~\cite{Bochmann:1994}, and they do not permit  nondeterministic choices between the application of inputs and the production of outputs. 
Several approaches have been investigated for  FSM-based test generation with guaranteed fault coverage~\cite{Bochmann:1994,YANNAKAKIS1995}.
FSMs have been extended to express time constraints, which  has resulted in a variety of timed FSMs ~\cite{MERAYO2008,2014arXiv1408.5967B,El-Fakih:2009,ZhigulinYMC11,ELFAKIH2014}. Timed FSMs are not compared to the well-known timed automata~\cite{ALUR1994} for which testing approaches have been developed~\cite{Krichen2009,DSSOULI201795,Nguena13} . 
Testing approaches~\cite{El-Fakih:2009,ZhigulinYMC11,ELFAKIH2014,Tvardovskii2017} for timed FSMs integrate the reasoning on time constraints in well-known FSM-based testing approaches~\cite{YANNAKAKIS1995,Broy:2005}. 
The work ~\cite{Hanh:2014} evaluate the application of different meta-heuristic algorithms to detect mutants of Simulink models. Meta-heuristic algorithms do not guarantee the  detection of all pre-defined nonconforming mutants.  %
The  methods in ~\cite{Petrenko2016a,Petrenko2016b} to verify and generate a complete test suite for FSM specifications are based solving constraints or Boolean expressions, which allows to take advantage of the efficiency of constraint/SAT solvers~\cite{DPLL62,Soos2009}; this is a novelty as compared to the work  in~\cite{Simao2010} and the well-known approaches such as the W-method. 
The high efficiency of using constraint solving  in testing software code was demonstrated in ~\cite{Gotlieb1998}. The constraints specify  the mutants surviving given tests; they are defined over the transitions in  executions of the mutants. 
The executions are selected with a so-called distinguishing automaton of the specification and the fault domain that is  compactly modeled with a nondeterministic FSM called a mutation machine. 
A solution of the constraints is a mutant which, if it is nonconforming,  allows to generate a test detecting the mutant and many others; then the constraints are  upgraded to generate new tests. 
\par Our contribution is to lift the methods in ~\cite{Petrenko2016a,Petrenko2016b} for verifying and generating complete test suites for fault domains for TFSM-T specifications. In our work, specifications and mutants are deterministic and input-complete TFSMs-T. 
We define a new distinguishing automaton with timeouts for a TFSM-T specification and a fault domain. The automaton serves to extract  transitions in detected mutants and build constraints for specifying test-surviving mutants. 
Extracting the transitions, we pair input/output transitions with timeout-unexpired transitions allowing to pass the input/output transitions; this is formalized with a notion of "comb". 
We have implemented the methods in a prototype tool which we use to evaluate the efficiency of the methods and compare our results with those of the related work.

\paragraph{Organization of the paper.}  The next section 
introduces a fault model for TFSMs-T and 
the coverage of  fault models with complete test suites. 
In  Section~\ref{sec:constraint} we build constraints  
for the analysis of timed input sequences and the 
generation of complete test suites. The analysis and 
generation methods are presented in Section~\ref{sec:analysis:generation}. 
Section~\ref{sec:exp} presents an empirical evaluation 
of the efficiency of the methods with the prototype tool. We conclude the paper in Section~\ref{sec:conc}.
\section{Preliminaries}
\par  Let  $\mathbb{R}_{\geq 0}$ and $\mathbb{N}_{\geq 1}$ denote the sets of non-negative real numbers and non-null natural numbers, respectively.   

\subsection{TFSM with timeouts}
\begin{definition}[TFSM with Timeouts] \label{def:TFSM} A {\em timed finite state machine with Timeouts}  (shortly, TFSM-T)
is a 6-tuple 
$\mathcal{S} = (S, s_0, I, O, \lambda_S, \Delta_\mathcal{S})$ 
where $S$, $I$ and  $O$ are  finite non-empty set of {\em states}, 
{\em inputs} and {\em outputs}, respectively,  $s_0$ is the 
{\em initial} state, $\lambda_\mathcal{S} \subseteq S \times I \times O \times S$ 
is an input/output transition relation and  
$\Delta_\mathcal{S} \subseteq S \times (\mathbb{N}_{\geq 1} \cup \{\infty\}) \times S$ 
is a timeout transition relation defining at least one timeout transition in every state. 
\end{definition}

\par 
Our definition of TFSM-T extends the definition in~\cite{2014arXiv1408.5967B} by allowing  multiple timeout transitions  in the same state, which we use later to compactly represent sets of TFSMs-T.  An input/output  transition $(s,i,o,s')\in \lambda_\mathcal{S}$ defines 
the output $o$ produced in its source state $s$ when input $i$ is applied.  A timeout transition $(s,\delta,s') \in \Delta_\mathcal{S}$ defines the timeout $\delta$ in state $s$.
A timeout transition can be taken if no input is applied at the current state before the timeout of the transition expires. It is not possible to wait for an input beyond the maximal timeout defined in the current state. 
%


%

\par A TFSM-T uses a single clock  
for recording  the time elapsing in the states and determining when timeouts expire. The clock is reset when the transitions are passed.
A {\em timed state}  of TFSM-T $\mathcal{S}$ is a pair 
$(s, x) \in S\times \mathbb{R}_{\geq 0}$ 
where $s \in S$ is a state of $\mathcal{S}$ and $x \in \mathbb{R}_{\geq 0}$ 
is the current value of the clock and $x < \delta$ for 
some $\delta \in \mathbb{N}_{\geq 1} \cup \{\infty\}$ 
such that $(s, \delta, s') \in \Delta_\mathcal{S}$. 
An {\em execution step} of $\mathcal{S}$ in timed state $(s,x)$ corresponds either to the time elapsing  or the passing of an input/output or  timeout transition; it is  {\em permitted} by a transition of $\mathcal{S}$. Formally,  $stp=(s,x)a(s',x') \in (S\times \mathbb{R}_{\geq 0}) \times ((I\times O) \cup \mathbb{R}_{\geq 0}) \times (S\times \mathbb{R}_{\geq 0})$ is an execution step  if it satisfies one of the  following conditions:
\begin{itemize}
\item (timeout step) $a \in \mathbb{R}_{\geq 0} $,  $x' = 0$    and $x + a = \delta$ for some $\delta$ such that $(s,\delta,s') \in \Delta_\mathcal{S}$; then $(s,\delta,s')$ is said to  permit the step.
\item (time-elapsing step) $a \in \mathbb{R}_{\geq 0} $,  $x'=x + a$, $s' =s$ and $x + a < \delta$
for some $\delta$ and $s'' \in S$ such that there exists $(s,\delta,s'') \in \Delta_\mathcal{S}$; then $(s,\delta,s'')$ is said to permit the step.
\item (input/output step) $a = (i,o)$ with $(i,o) \in I \times O$, $x' =0$ and there exists  $(s,i,o,s') \in \lambda_\mathcal{S}$; then $(s,i,o,s')$ is said to permit the step.
\end{itemize}
Time-elapsing steps satisfy the following time-continuity property w.r.t the same timeout transition:
if $(s_1,x_1)d_1(s_2,x_2)d_2(s_3,x_3)\ldots d_{k-1}\linebreak(s_k,x_k)$ is a sequence of time-elapsing steps permitted by the same timeout transition $t$, then $(s_1,x_1)d_1+d_2+\ldots +d_{k-1}(s_k,x_k)$ is a time-elapsing step permitted by $t$. In the sequel, any time-elapsing step permitted by a timeout transition can be represented with a sequence  of time-elapsing steps permitted by the same timeout transitions. 
\par  An {\em execution} of  $\mathcal{S}$ in timed state $(s,x)$ is a sequence of steps $e=stp_1stp_2\ldots stp_n$ with $stp_{k} =(s_{k-1},x_{k-1})a_{k}(s_{k},x_{k})$, $k\in [1,n]$ such that the following conditions hold:
\begin{itemize}
\item $(s_0,x_0) =(s,x)$,    
\item $stp_1$ is not an input/output step,
\item $stp_{k}$ is a input/output step implies that $stp_{k-1}$ is a time-elapsing step for every $k\in [1..n]$.
\end{itemize}
If needed, the elapsing of zero time units can be inserted between a timeout step and an input/output step.
%
%
\newcommand{\trname}[1]{$^{[\boldsymbol{#1}]}$}
\newcommand{\runame}[1]{$^{(\boldsymbol{#1})}$}
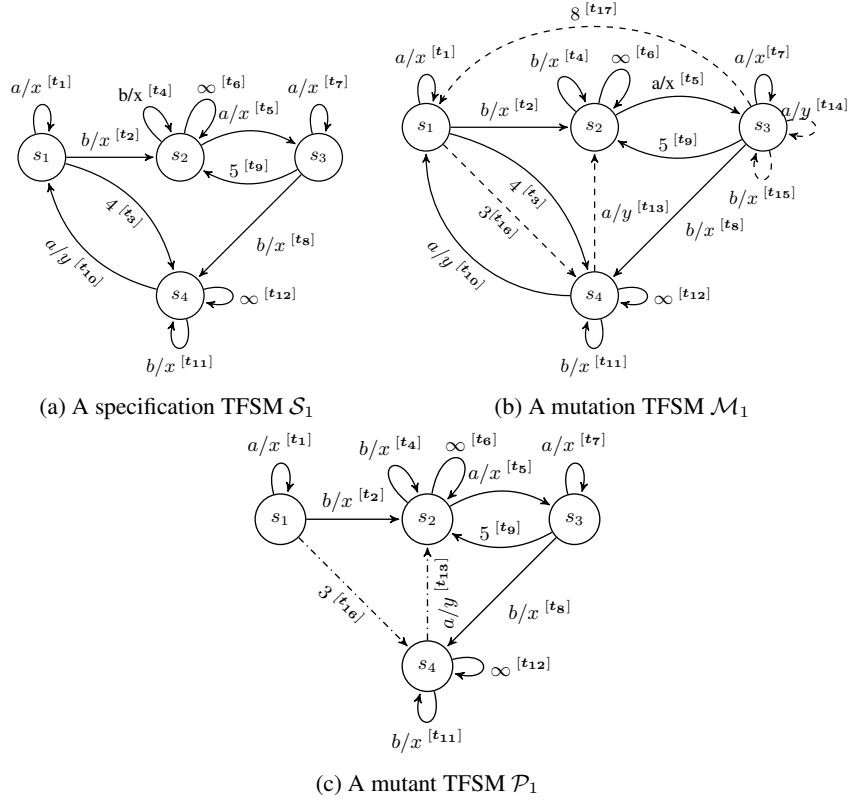
\begin{figure*}[t]
  \centering
  \subfloat[A specification TFSM $\mathcal{S}_1$]{
     \label{fig:spec}
     \scalebox{0.80}{     
         \begin{tikzpicture}[->,>=stealth',shorten >=1pt,auto,node distance=2.3cm, semithick]
    \node[state] (1)                    {$s_1$};
    \node[state]         (2) [right of=1] {$s_2$};
    \node[state]         (3) [right of=2] {$s_3$};
    \node[state]         (4) [below of=2] {$s_4$};

    \path (1) edge [loop above] node[midway, above] {$a/x$ \trname{t_1}}         (1)
              edge              node[midway, above] {$b/x$ \trname{t_2}}         (2)
              edge [bend left]  node[sloped,below] {$4$ \trname{t_3}}    (4)
          (2) edge [loop , in=110, out=140, looseness=10] node[midway, above] {b/x  \trname{t_4}}         (2)
              edge [bend left]  node[midway, above] {$a/x$ \trname{t_5}}         (3)
              edge [loop, in=45, out=80, looseness=10]	node[above, xshift=0.2cm] {$\infty$ \trname{t_6}}         (2)
          (3) edge [loop above] node[midway, above] {$a/x$ \trname{t_7}}         (3)
              edge              node[midway, below right] {$b/x$ \trname{t_8}}   (4)
              edge [bend left]	node[midway, above] {$5$ \trname{t_9}}         (2)
          (4) edge [bend left]  node[sloped, below] {$a/y$ \trname{t_{10}}} (1)
              edge [loop below] node[midway, below] {$b/x$ \trname{t_{11}}}       (4)
              edge [loop right] node[midway, right] {$\infty$ \trname{t_{12}}}      (4);
  \end{tikzpicture}
       }
     }\hspace{2pt}
  \subfloat[A mutation TFSM $\mathcal{M}_1$]{
    \label{fig:mutmachine}
    \scalebox{0.80}{ 
        \begin{tikzpicture}[->,>=stealth',shorten >=1pt,auto,node distance=2.8cm, semithick]
    \node[state] (1)                    {$s_1$};
    \node[state]         (2) [right of=1] {$s_2$};
    \node[state]         (3) [right of=2] {$s_3$};
    \node[state]         (4) [below of=2] {$s_4$};

    \path (1) edge [loop above] node[midway, above] {$a/x$ \trname{t_1}}         (1)
              edge              node[midway, above] {$b/x$ \trname{t_2}}         (2)
              edge [bend left]  node[sloped, below] {$4$ \trname{t_3}}    (4)
              edge [dashed]     node[sloped, below] {$3$\trname{t_{16}}}    (4)
          (2) edge [loop , in=110, out=140, looseness=10] node[midway, above] {$b/x$ \trname{t_4}}         (2)
              edge [bend left]  node[midway, above] {a/x \trname{t_5}}         (3)
              edge [loop, in=45, out=80, looseness=10]	node[above, xshift=0.2cm] {$\infty$ \trname{t_6}}         (2)
          (3) edge [loop above] node[midway, above] {$a/x$\trname{t_7}}         (3)
              edge [dashed, loop right] node[midway, above] {$a/y$ \trname{t_{14}}} (3)
              edge [dashed, loop below] node[midway, below] {$b/x$ \trname{t_{15}}} (3)
              edge              node[midway, below right] {$b/x$ \trname{t_8}}   (4)
              edge [bend left]	node[midway, above] {$5$ \trname{t_9}}         (2)
              edge [dashed, bend right=60]	node[midway, above] {$8$ \trname{t_{17}}}         (1)
          (4) edge [bend left=45]  node[sloped, below] {$a/y$ \trname{t_{10}}} (1)
              edge [dashed]     node[midway, right] {$a/y$ \trname{t_{13}}}      (2)
              edge [loop below] node[midway, below] {$b/x$ \trname{t_{11}}}      (3)
              edge [loop right] node[midway, right] {$\infty$ \trname{t_{12}}}      (0);
  \end{tikzpicture}
     }
    }\hspace{2pt}
  \subfloat[A mutant TFSM $\mathcal{P}_1$]  { 
  \label{fig:mutant}
   \scalebox{0.85}{
    \begin{tikzpicture}[->,>=stealth',shorten >=1pt,auto,node distance=2.3cm, semithick]
    \node[state] (1)                    {$s_1$};
    \node[state]         (2) [right of=1] {$s_2$};
    \node[state]         (3) [right of=2] {$s_3$};
    \node[state]         (4) [below of=2] {$s_4$};

    \path (1) edge [loop above] node[midway, above] {$a/x$ \trname{t_1}}         (1)
              edge              node[midway, above] {$b/x$ \trname{t_2}}         (2)
              edge [dashdotted]  node[sloped,below] {$3$ \trname{t_{16}}}    (4)
          (2) edge [loop , in=110, out=140, looseness=10] node[midway, above] {$b/x$  \trname{t_4}}         (2)
              edge [bend left]  node[midway, above] {$a/x$ \trname{t_5}}         (3)
              edge [loop, in=45, out=80, looseness=10]	node[above, xshift=0.2cm] {$\infty$ \trname{t_6}}         (2)
          (3) edge [loop above] node[midway, above] {$a/x$ \trname{t_7}}         (3)
              edge              node[midway, below right] {$b/x$ \trname{t_8}}   (4)
              edge [bend left]	node[midway, above] {$5$ \trname{t_9}}         (2)
          (4) edge [dashdotted]  node[sloped, below] {$a/y$ \trname{t_{13}}} (2)
              edge [loop below] node[midway, below] {$b/x$ \trname{t_{11}}}       (4)
              edge [loop right] node[midway, right] {$\infty$ \trname{t_{12}}}      (4);
  \end{tikzpicture}
   }
   }
  \caption{Examples of TFSMs, state $s_1$ is initial. Dashed arrows represent mutated transitions and solid arrows represent transitions of the specification. Names of transitions appear in brackets.}
  \label{fig:tfsm}
\end{figure*}
Let $d_1 d_2 \ldots d_l \in \mathbb{R}_{\geq 0}^l$ be a  non-decreasing sequence of real numbers, i.e., 
$d_k \leq d_{k+1}$ for every $k=1..l-1$. 
The sequence  $\sigma_e = ((i_1,o_1), d_1)((i_2,o_2), d_2)\ldots((i_l,o_l), d_l)$ in $((I\times O)\times \mathbb{R}_{\geq 0})^*$ with $l<n$ is a {\em timed input/output sequence} of execution $e$ if $(i_1,o_1)(i_2,o_2)\ldots(i_l,o_l)$ is the maximal sequence of input/output pairs occurring in $e$. 
The delay $d_k$ for each input/output pair $(i_k,o_k)$, with $k=1..l$, is  the amount of the time elapsed from  the beginning of $e$ to the occurrence of $(i_k,o_k)$. 
The {\em timed input sequence} and the {\em timed output sequence} of $e$ are   $(i_1, d_1)(i_2, d_2)...(i_l, d_l)$ and $(o_1, d_1)(o_2, d_2)...(o_l, d_l)$, respectively. We let  $inp(e)$ and $out(e)$  denote the timed input and output sequences of execution $e$.
Given a timed input sequence $\alpha$, let $out_\mathcal{S}((s,x), \alpha)$ denote the set of all timed output sequences which can be produced by $\mathcal{S}$ when  $\alpha$ is applied in $s$, i.e.,  $out_\mathcal{S}((s,x), \alpha) =\{ out(e)~\mid~e$ is  an  execution of   $\mathcal{S} \text{  in  } ( s,x) \text{ and } inp(e) = \alpha\}$.

\par A TFSM-T $\mathcal{S}$ is {\em deterministic} (DTFSM-T) if it defines at most one input-output transition for each tuple $(s, i) \in S \times I$ and exactly one timeout transition in each state; otherwise, it is {\em nondeterministic}. 
$\mathcal{S}$ is {\em initially connected} if it has an execution from its initial state to each of its states. 
$\mathcal{S}$ is {\em complete} if for each tuple $(s, i) \in S \times I$ it defines at least one input-output transition. Note that the set of timed input sequences defined in each state of a complete machine $\mathcal{S}$ is $(I \times \mathbb{R}_{\geq 0})^*$.

\par We define distinguishability and equivalence relations between states of complete TFSMs-T. Similar notions were introduced in~\cite{ZhigulinYMC11}. Intuitively, states producing different timed output sequences in response to the same timed input sequence are
distinguishable. 
Let $p$ and $s$ be the states of two complete TFSMs-T over the same inputs and outputs. Given a timed input sequence $\alpha$, $p$ and $s$ are {\em distinguishable} (with
distinguishing input sequence $\alpha$), denoted $p \not\simeq_\alpha s$, if the sets of timed output sequences in 
$out_\mathcal{S}((p,0),\alpha)$ and $out_\mathcal{S}((s,0),\alpha)$ differ; otherwise they are {\em
equivalent} and we write $s \simeq p$, i.e., if the sets of timed output sequences coincide for all timed input sequence $\alpha$.

\begin{definition}[Submachine]  TFSM-T $\mathcal{S} = (S, s_0, I, O, \lambda_\mathcal{S}, \Delta_\mathcal{S})$ is a {\em submachine} of  TFSM-T  $\mathcal{P} = (P, p_0, I, O, \lambda_\mathcal{P}, \Delta_\mathcal{P})$ if   $S \subseteq P$, $s_0 = p_0$, $\lambda_\mathcal{S} \subseteq \lambda_\mathcal{P}$ and $\Delta_\mathcal{S} \subseteq \Delta_\mathcal{P}$.
\end{definition}

\begin{example}\label{ex:example1}
Figure~\ref{fig:tfsm} presents two initially connected TFSMs-T $\mathcal{S}_1$ and $\mathcal{M}_1$. $\mathcal{M}_1$ is nondeterministic; it defines two timeout transitions in states $s_1$ and $s_3$, which is not allowed in~\cite{2014arXiv1408.5967B}. $\mathcal{S}_1$ is  a complete deterministic submachine of $\mathcal{M}_1$. 
\par Here are four executions of the TFSM-T $\mathcal{M}_1$ in Figure~\ref{fig:mutmachine}, where the transitions defining the steps appear below the arrows:
\begin{enumerate}
	\item $(s_1, 0)\xrightarrow[t_3 | t_{16}]{2}(s_1, 2)\xrightarrow[t_2]{b, x}(s_2, 0)\xrightarrow[t_6]{1}(s_2, 1)\xrightarrow[t_5]{a, x}(s_3, 0)\xrightarrow[t_{17}]{8}(s_1, 0)\xrightarrow[t_3]{4}(s_4, 0)\xrightarrow[t_{12}]{0.5}(s_4, 0.5)\xrightarrow[t_{10}]{a, y}(s_1, 0)$
    \item $(s_1, 0)\xrightarrow[t_3]{4}(s_4, 0)\xrightarrow[t_{12}]{0.5}(s_4, 0.5)\xrightarrow[t_{10}]{a, y}(s_1, 0)\xrightarrow[t_{16}]{3}(s_4, 0)\xrightarrow[t_{12}]{0.7}(s_4, 0.7)\xrightarrow[t_{13}]{a, y}(s_2, 0)$
    \item $(s_1, 0)\xrightarrow[t_3]{3.5}(s_1, 3.5)\xrightarrow[t_2]{b, x}(s_2, 0)\xrightarrow[t_6]{1}(s_2, 1)\xrightarrow[t_5]{a, x}(s_3, 0)\xrightarrow[t_{17}]{8}(s_1, 0)\xrightarrow[t_3]{4}(s_4, 0)\xrightarrow[t_{12}]{0.5}(s_4, 0.5)\xrightarrow[t_{10}]{a, y}(s_1, 0)$
    \item $(s_1, 0)\xrightarrow[t_{16}]{3}(s_4, 0)\xrightarrow[t_{12}]{0.5}(s_4, 0.5)\xrightarrow[t_{11}]{b, x}(s_4, 0)\xrightarrow[t_{12}]{1}(s_4, 1)\xrightarrow[t_{13}]{a, y}(s_2, 0)\xrightarrow[t_{6}]{12.5}(s_2, 12.5)\xrightarrow[t_5]{a,x}(s_3, 0)$
\end{enumerate} 
Let us explain the first execution.  It has $8$ steps represented with arrows between timed states. The label above an arrow is either a delay in ${\mathbb{R}_{\geq 0}}$ or an input-output pair. The label below an arrow indicates the transitions permitting the step. The first, third and seventh steps of the first execution are time-elapsing. The second, fourth and last steps are input-output.  The fifth and the sixth steps are timeout. The first step is permitted by either transition $t_3$ or $t_{16}$ because their timeouts are not expired $2$ units after the machine has entered state $(s_1,0)$. The timeout of transition $t_{12}$ permitting the seventh step has not expired before  input-output transition $t_{10}$ is performed at the last step. The difference between the first and the third execution is that input $b$ is applied lately in the third execution, i.e., $3.5$ time units after the third execution has started. 
The timed input/output sequences for the four executions are   
$((b,x), 2)((a,x), 3)((a,y), 15.5)$, $((a,y), 4.5)((a,y), 8.2)$, $((b,x), 3.5)((a,x), 4.5)((a,y), 17)$ and $\linebreak((b,x), 3.5)((a,y), 4.5)((a,x), 17)$, respectively. The timed input sequence and the timed output sequence for the first execution are  $(b, 2)(a, 3)(a, 15.5)$ and $(x, 2)(x, 3) (y, 15.5)$. Similarly, we can determine the timed input and output sequences for the three other executions. The third and the fourth executions have the same timed input sequence but different timed output sequences.\\
\end{example}

Henceforth the TFSMs-T are complete and initially connected.

\subsection{Complete test suite for fault models}
 Let $\mathcal{S} = (S, s_0, I, O, \lambda_\mathcal{S}, \Delta_\mathcal{S})$ be a DTFSM-T, called the {\em specification} machine.
\begin{definition}[Mutation machine for a specification machine] A nondeterministic TFSM-T $\mathcal{M} = (M, m_0, I, O, \lambda_\mathcal{M}, \Delta_\mathcal{M})$ is a {\em mutation machine} of   $\mathcal{S}$ if $\mathcal{S}$ is a submachine of $\mathcal{M}$. Transitions in   $\lambda_\mathcal{M}$ but not in  $\lambda_\mathcal{S}$ or in $\Delta_\mathcal{M}$ but not in $\Delta_\mathcal{S}$ are called {\em mutated}.
\end{definition}
\newcommand{\Mut}{\textit{Mut}} 
A mutant is a deterministic submachine of $\mathcal{M}$ different from the specification. We let $\Mut(\mathcal{M})$ denote the set of mutants in $\mathcal{M}$. A mutant represents an implementation of the specification seeded with  faults. Faults are represented with mutated transitions  and every mutant defines a subset of them. Mutated transitions can represent transfer faults, output faults, changes of timeouts and adding of extra-states. 
\par A transition $t$   is {\em suspicious} in $\mathcal{M}$ if $\mathcal{M}$ defines another transition $t'$ from the source state of $t$ and either both $t$ and $t'$  have the same input or they are timeout transitions. A transition of the specification is called {\em untrusted} if it is suspicious in the mutation machine; otherwise, it is {\em trusted}. The set of suspicious transitions of $\mathcal{M}$ is partitioned into a set of untrusted transitions all defined in the specification and the set of mutated transitions undefined in the specification. \\

\par Let $\mathcal{P}$ be a mutant with an initial state $p_0$ of the mutation machine $\mathcal{M}$ of $\mathcal{S}$. We use the state equivalence relation $\simeq$ to define conforming mutants. 
\begin{definition}[Conforming mutants and detected mutants]\label{def:conforming}
Mutant $\mathcal{P}$  is {\em conforming} to  $\mathcal{S}$, if $p_0 \simeq s_0$; otherwise, it is {\it nonconforming} and a timed input sequence $\alpha$ such that $out_\mathcal{P}((p_0,0),\alpha) \neq out_\mathcal{S}((s_0,0),\alpha)$ is said to {\em detect} $\mathcal{P}$.
\end{definition}
We say that mutant $\mathcal{P}$ {\em survives} input sequence $\alpha$ if $\alpha$ does not detect $\mathcal{P}$.

 \par The set  $\Mut(\mathcal{M})$ of all mutants in mutation machine $\mathcal{M}$ is called a {\em fault domain} for $\mathcal{S}$. If $\mathcal{M}$ is deterministic and complete then  $\Mut(\mathcal{M})$ is empty. A general {\em fault model} is the tuple $\langle
\mathcal{S}, \simeq, \Mut(\mathcal{M})\rangle$ following~\cite{PetrenkoY92,Nguena17}. 
Let $\lambda_\mathcal{M}(s,i)$ denote the set of input/output transitions defined in state $s$ with input $i$ and $\Delta_\mathcal{M}(s)$ denote the set of timeout transitions defined in state $s$. The number of mutants in  $\Mut(\mathcal{M})$ is given by the formula
$ \displaystyle |\Mut(\mathcal{M})| = \Pi_{(s,i)\in S\times I}|\lambda_\mathcal{M}(s,i)| \times \Pi_{s \in S}|\Delta_\mathcal{M}(s)|-1 $,
where $\lambda_\mathcal{M}(s,i)$ denotes the set of input-output transitions with input $i$ defined in $s$ and $\Delta_\mathcal{M}(s)$ denotes the set of timeout transitions  in $s$.
The conformance relation
partitions the set $\Mut(\mathcal{M})$ into conforming mutants and nonconforming ones which  we need to detect.
\begin{definition}[Complete test suite]
A test for $\langle
\mathcal{S}, \simeq, \Mut(\mathcal{M})\rangle$ is a timed input sequence. A {\em complete test suite} for  $\langle
\mathcal{S}, \simeq, \Mut(\mathcal{M})\rangle$ is a set of test detecting all nonconforming mutants in $\Mut(\mathcal{M})$. 
\end{definition}

\begin{example}
\par  In Figure~\ref{fig:tfsm}, $\mathcal{M}_1$ is a mutation machine  for the specification machine $\mathcal{S}_1$. $\mathcal{M}_1$ and  $\mathcal{S}_1$ has the same number of states, meaning that the faults represented with mutated transitions in  $\mathcal{M}_1$ do not introduce extra-states.  The mutated transitions are represented by dashed lines. The transitions $t_3$, $t_{14}$ and $t_7$ are  suspicious; however  $t_1$ is not. $t_1$ is trusted and $t_7$ is untrusted. $t_{14}$ is neither trusted nor untrusted because it does not belong to the specification. $t_{14}$ defines an output fault on input $a$ since the expected output is defined with transition $t_7$. In state $s_3$, $t_{15}$ defines a transfer fault for input $b$ and  $t_{17}$ increases the expected timeout for $s_3$ and defines a transfer fault. The transition $t_{16}$ implements a fault created by reducing the timeout of $t_3$; it is defined in the mutant $\mathcal{P}_1$ in Figure~\ref{fig:mutant}.  For the timed input sequence $(b,3.5)(a,4.5)(a,17)$, the specification $\mathcal{S}_1$ and mutant performs the third and fourth executions in Example~\ref{ex:example1}, respectively.  $\mathcal{P}_1$ is nonconforming because the produced  timed input sequence $(x,3.5)(y,4.5)(x,17)$ differ from $(x,3.5)(x,4.5)(y,17)$, the timed output sequence produced by   $\mathcal{S}_1$. $\mathcal{M}_1$ defines $31$ mutants; some of them are conforming and we would like to generate a complete test suite detecting all the nonconforming mutants.
\end{example}

To generate a complete test suite, a test can be computed for each nonconforming mutant by enumerating the mutants one-by-one, which would be inefficient for huge fault domains. We avoid the one-by-one enumeration of the mutants with constraints specifying only test-surviving mutants.

\section {Specifying test-surviving mutants \label{sec:constraint}}
  The mutants surviving a test  cannot produce any execution with an unexpected  timed output sequence for the test. We encode them with Boolean formulas over Boolean transition variables of which the values indicate the presence or absence of transitions of the mutation machine in mutants. 
\subsection{Revealing combs and involved mutants}
The mutants detected  by a test $\alpha$ exhibit a revealing execution which produces an unexpected timed output sequence and has $\alpha$ as the test. Such an execution is permitted by transitions forming a comb-subgraph in the state transition diagram of the mutation machine. Intuitively,  a comb for an execution is nothing else but a path  augmented with  timeout-unexpired transitions, i.e., transitions of which the timeouts have not expired prior to performing an input-output step. These additional timeout transitions are also needed to specify detected mutants and eliminate them from the fault domain. To simplify the notation, we represent comb-subgraphs with sequences of transitions.

\par A {\em comb} of an execution $e=stp_1stp_2\ldots stp_n$ is the sequence of transitions $t_1t_2\cdots t_n$ such that  $t_i$ permits  $stp_i$ for every $i=1..n$. We say that comb $t_1t_2\cdots t_n$ is {\em enabled} by the input sequence of $e$. Each timeout  or input/output step in $e$ is permitted with a unique transition. However, each time-elapsing step is permitted by a  timeout transition with an unexpired timeout,  i.e., the timeout is not greater than the clock value in the source timed state of the  time-elapsing step. So, several combs can permit the same execution since  several timeout transitions permit the same time-elapsing step. Note that timeout transitions with finite or infinite timeouts appear in combs when they permit time-elapsing steps preceding input/output steps; later such timeout transitions participate in Boolean encodings of combs involving  detected mutants. 
 
\begin{figure}[t]
\centering
\scalebox{0.7}{\begin{tikzpicture}[->,>=stealth',shorten >=1pt,auto,node distance=2.3cm, semithick]
    \node[state] (1)                    {$s_1$};
    \node[state]         (2) [below of=1] {$s_4$};
    \node[state]         (3) [right of=1] {$s_2$};
    \node[state]         (4) [below of=3] {$s_2$};
    \node[state]         (5) [right of=3] {$s_3$};
    \node[state]         (6) [right of=5] {$s_1$};
    \node[state]         (7) [right of=6] {$s_4$};
    \node[state]         (8) [below of=7] {$s_4$};
    \node[state]         (9) [right of=7] {$s_1$};

    \path (1) edge  node[midway, right] {$4$ \trname{t_3}}         (2)
    		  edge  node[midway, above] {$b/x$ \trname{t_2}}         (3)
          (3) edge  node[midway, right] {$\infty$  \trname{t_6}}         (4)
              edge  node[midway, above] {$a/x$ \trname{t_5}}         (5)
          (5) edge  node[midway, above] {$8$ \trname{t_{17}}} (6)
          (6) edge  node[midway, above] {$4$ \trname{t_{3}}} (7)
          (7) edge  node[midway, right] {$\infty$ \trname{t_{12}}} (8)
              edge  node[midway, above] {$a/y$ \trname{t_{10}}} (9);
  \end{tikzpicture}
  }
\caption{A comb for the first execution in Example~\ref{ex:example1}. \label{fig:comb}}
\end{figure}
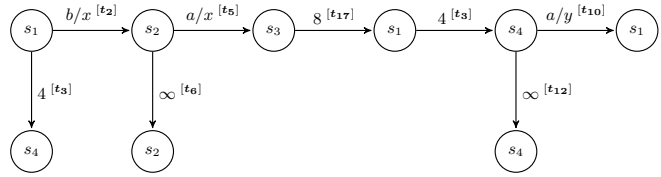

\begin{example} There are two combs for the first execution in Example~\ref{ex:example1}; this is because the first step of the execution is permitted either by $t_3$ or $t_{16}$.  The first comb $t_3  t_2  t_6  t_5  t_{17}  t_{3}  t_{12}  t_{10}$ is represented in Figure~\ref{fig:comb}. The timeouts of the transitions represented with vertical arrows have not expired in the execution when the input-output transition is performed. The timeouts of the transitions represented with horizontal arrows have expired. The first comb  is deterministic  whereas the second comb $t_{16}  t_2  t_6  t_5  t_{17}  t_{3}  t_{12}  t_{10}$ is nondeterministic because  $t_{16}$ and $t_3$ are two suspicious timeout  transitions defined in $s_1$. The first comb for the first execution is also a comb for the third execution in Example~\ref{ex:example1}; but the second comb is not a comb for the third execution. This is because the application of $b$ after $3.5$ time units in $s_1$ is possible if the timeout transition $t_{16}$ is not passed.  The second execution corresponds to a single nondeterministic comb $t_3 t_{12} t_{10}  t_{16} t_{12} t_{13}$. It is nondeterministic because   $t_{10}$ and $t_{13}$ are two input/output transitions defined in $s_4$ with the same input $a$. Each occurrence of the timeout transition $t_{12}$ before an input/output transition indicates that the timeout of $t_{12}$ has not expired before the input/output transition is passed.
\end{example} 

\par Combs for executions with unexpected timed output sequences reveal nonconforming mutants unless they belong only to nondeterministic submachines. The combs belonging  only to nondeterministic submachines have two transitions which are not defined in the same mutant; such combs are called  nondeterministic.
\begin{definition}[Deterministic and nondeterministic combs]
A comb is {\em nondeterministic} if  it has two suspicious input-output transitions or two timeout transitions defined in an identical state of the mutation  machine; otherwise, it is a {\em deterministic comb}.
\end{definition}

Clearly, combs in a mutant or the specification are deterministic because two suspicious transitions cannot be defined in an identical state in a mutant or the specification. A nondeterministic submachine of a mutation machine can contain both deterministic and nondeterministic combs.
%
\begin{figure*}[t]
\centering
\scalebox{0.7}{\begin{tikzpicture}[->,>=stealth', state/.style={draw,rectangle},shorten >=1pt,auto,node distance=3.0cm, semithick]
    \node[state] (1)              {$s_1,s_1,0,0$};
    \node[state]         (2) [right of=1] {$s_2,s_2,0,0$};
    \node[state]         (3) [right of=2] {$s_3,s_3, 0, 0$};
    \node[state]         (4) [right of=3] {$s_2, s_3, 0, 5$};
    \node[state]         (5) [right of=4] {$s_2, s_4, 0, 0$};
    \node[state]         (10) [below right of=4] {$s_2, s_1, \infty, 0$};
    \node[state]         (6) [below left of=10] {$s_2, s_3, 0, 0$};
    \node[accepting, state]         (7) [right of=6] {$\nabla$};
    \node[state]         (9) [below of=1] {$s_4, s_4, 0, 0$};
    \node[state]         (8) [right of=9] {$s_1, s_4, 3, 0$};
    \node[state]         (11) [left of=6] {$s_2,s_2,\infty, 0$};

    \path (1) edge              node[midway, above] {$b$ \trname{t_2} \runame{\mathcal{R}_1}}         (2)
              edge      node[sloped, below] {$3$ \trname{t_{16}} \runame{\mathcal{R}_4}}   (8)
              edge      node[sloped, below] {$4$ \trname{t_3} \runame{\mathcal{R}_3}}   (9)
          (2) edge [loop above] node[midway, above] {$b$ \trname{t_4} \runame{\mathcal{R}_1}}         (2)
              edge [loop below] node[midway, below] {$\infty$ \trname{t_6} \runame{\mathcal{R}_3}}         (2)
              edge [bend left]  node[midway, above] {$a$ \trname{t_5} \runame{\mathcal{R}_1}}         (3)
          (3) edge      node[midway, above] {$5$ \trname{t_{17}} \runame{\mathcal{R}_6}} (4)
              edge [bend left]	node[midway, below] {$5$ \trname{t_9} \runame{\mathcal{R}_3}}         (2)
          (4) edge              node[midway, above] {$b$ \trname{t_{8}} \runame{\mathcal{R}_1}}      (5)
              edge              node[below,sloped] {$b$ \trname{t_{15}} \runame{\mathcal{R}_1}}      (6)
              edge      node[midway, above right] {$3$ \trname{t_{17}} \runame{\mathcal{R}_5}}      (10)
          (5) edge              node[midway, right] {$a$ \trname{t_{10}} \runame{\mathcal{R}_2}}      (7)
              edge [bend left=100]  node[midway,right] {$a$ \trname{t_{13}} \runame{\mathcal{R}_2}}      (7)
              edge [loop above] node[midway, above] {$\infty$ \trname{t_{12}} \runame{\mathcal{R}_3}}         (5)
          (6) edge              node[midway, below left] {$a$ \trname{t_{6}} \runame{\mathcal{R}_1}}      (3)
              edge              node[midway, below] {$a$ \trname{t_{14}} \runame{\mathcal{R}_2}}      (7)
              edge       node[midway, below right] {$8$ \trname{t_{17}} \runame{\mathcal{R}_5}}      (10)
              edge      node[midway, below] {$5$ \trname{t_{9}} \runame{\mathcal{R}_5}}      (11);
  \end{tikzpicture}}
\caption{A fragment of the distinguishing automaton of $\mathcal{S}_1$ and $\mathcal{M}_1$ with timeouts\label{fig:da}; $(s_1,s_1,0,0)$ is initial.}
\end{figure*}
%
\begin{definition}[Revealing comb]\label{def:rev:comb}
Let $\pi$ be a comb of an execution $e_1$ from $(s_0,0)$. We say that 
$\pi$ is {\em $\alpha$-revealing} if there exists  an execution $e_2$ of $\mathcal{S}$ such that $\alpha =inp(e_1) = inp (e_2)$, $out(e_1) \neq out(e_2)$ while this does not hold for any prefix of $\alpha$. Comb $\pi$  is revealing if it is $\alpha$-revealing for some input sequence $\alpha$. 
\end{definition}
The specification contains no revealing comb because its executions always produce expected timed output sequences. Only mutants, nondeterministic or incomplete submachines of a mutation machine can contain revealing combs; but mutants contains only deterministic revealing combs.

Let $Rev_\alpha(\mathcal{P})$ denote the set of deterministic $\alpha$-revealing combs of machine $\mathcal{P}$. 

\begin{lemma}~\label{lemma:comb:equal}
$Rev_\alpha(\mathcal{M}) = \bigcup_{\mathcal{P} \in \Mut(\mathcal{M})} Rev_\alpha(\mathcal{P})$, for any test $\alpha$.
\end{lemma}
\begin{proof}
 Mutants are deterministic submachines of the mutation machine. Each deterministic revealing comb in $\mathcal{M}$  is a comb in a mutant, meaning that $Rev_\alpha(\mathcal{M}) \subseteq \bigcup_{\mathcal{P} \in \Mut(\mathcal{M})} Rev_\alpha(\mathcal{P})$. Combs in mutants are necessarily deterministic and revealing combs in mutants are deterministic revealing comb in the mutation machine because mutants are deterministic submachines of the mutation machine. Thus $\bigcup_{\mathcal{P} \in \Mut(\mathcal{M})} Rev_\alpha(\mathcal{P}) \subseteq  Rev_\alpha(\mathcal{M}) $
\end{proof}

\begin{lemma}\label{lemma:rev:detected}
Let $\pi \in Rev_\alpha(\mathcal{M})$ for a test $\alpha$ and $\mathcal{P} \in \Mut(\mathcal{M})$. 
$\mathcal{P}$ is not detected by  $\alpha$ if and only if $\pi \not\in Rev_\alpha(\mathcal{P})$.   
\end{lemma}
\begin{proof}
 Assume that  $\mathcal{P} \in \Mut(\mathcal{M})$ is not detected by $\alpha$. Then $\mathcal{P}$ is either conforming or nonconforming. $\mathcal{P}$ contains no revealing comb if it is conforming; then $\pi \not\in Rev_\alpha(\mathcal{P})$ for every $\pi \in Rev_\alpha(\mathcal{M})$. If $\mathcal{P}$ is nonconforming and there is $\pi \in Rev_\alpha(\mathcal{P}) \cap  Rev_\alpha(\mathcal{M})$, we get a contradiction with the fact that  $\mathcal{P}$ is not detected by $\alpha$ because  $\pi$ is the comb for  an execution in $\mathcal{P}$ with timed input sequence $\alpha$ and an unexpected timed output sequence. Conversely, by Definition~\ref{def:conforming} and Definition~\ref{def:rev:comb} if $Rev_\alpha(\mathcal{P}) \cap  Rev_\alpha(\mathcal{M}) =\emptyset$ then $\alpha$ does not detect $\mathcal{P}$
\end{proof}

\begin{example}
The comb for the fourth execution in Example~\ref{ex:example1}, namely $t_{16}t_{12}t_{11}t_{12}t_{13}t_{6}t_{4}$ is revealing and contained in the mutant and mutation machine in Figure~\ref{fig:tfsm}. To prevent the fourth execution, we must prevent one the transitions in the comb. For example, if we prevent $t_{16}$, the other timeout transition $t_3$ defined in $s_1$ will be performed, yielding to the third execution which cannot be performed in the mutant $\mathcal{P}_1$, but rather in  $\mathcal{S}_1$. Clearly $\mathcal{S}_1$ is not detected by $(b,3.5)(a,4.5)(a,17)$, in the contrary of $\mathcal{P}_1$. 
\end{example} 

We define a distinguishing automaton with timeouts for the specification and the mutation machine; the automaton  define all revealing combs from which we will serve to extract deterministic combs which reveal nonconforming mutants.

\begin{definition}[Distinguishing automaton with timeouts]
\label{def:da}
Given a specification machine $\mathcal{S} = (S, s_0, I, O, \lambda_\mathcal{S}, \Delta_\mathcal{S})$ and a mutation machine $\mathcal{M} = (M, m_0, I, O, \lambda_\mathcal{M}, \Delta_\mathcal{M})$, a finite automaton $\mathcal{D} = (C \cup \{\nabla\}, c_0, I, \lambda_\mathcal{D}, \Delta_\mathcal{D}, \nabla)$, where $C \subseteq S \times S \times (\mathbb{N}_{\geq 0} \cup \{\infty\}) \times (\mathbb{N}_{\geq 0} \cup \{\infty\})$, 
$\lambda_\mathcal{D}\subseteq C \times I\times C$ is the input transition  relation, $\Delta_\mathcal{D}\subseteq C\times(\mathbb{N}_{\geq 1} \cup \{\infty\}) \times C$ is the timeout transition relation and $\nabla$ is the accepting (sink) state, is the {\em distinguishing} automaton with timeouts for $\mathcal{S}$ and $\mathcal{M}$, if it holds that:
\begin{itemize}
\setlength\itemsep{0.5em}
	\item $c_0 = (s_0, m_0, 0, 0)$
    \item For each $(s, m, x_s, x_m) \in C$ and $i \in I$
    \begin{itemize}
    \setlength\itemsep{0.5em}
    	\item[$(\mathcal{R}_1)$]: ${((s, m, x_s, x_m), i, (s', m', 0, 0)) \in \lambda_\mathcal{D}}$ if there exists $(s, i, o, s') \in \lambda_\mathcal{S}, (m, i, o', m') \in \lambda_\mathcal{M}$ s.t.  $o = o'$
    	\item[$(\mathcal{R}_2)$]: ${((s, m, x_s, x_m), i, \nabla) \in \lambda_\mathcal{D}}$ if there exists $(s, i, o, s') \in \lambda_\mathcal{S}$, $(m, i, o', m') \in \lambda_\mathcal{M}$ s.t. $o \neq o'$
    \end{itemize}
    \item For each $(s, m, x_s, x_m) \in C$ and the only timeout transition $(s, \delta_s, s') \in \Delta_\mathcal{S}$ defined in the state of the deterministic specification
    \begin{itemize}
    \setlength\itemsep{0.5em}
    	\item[$(\mathcal{R}_3)$]: $((s, m, x_s, x_m), \delta_m  -  x_m, (s', m', 0, 0)) \in \Delta_\mathcal{D}$  if there exists  $(m, \delta_m, m') \in \Delta_\mathcal{M}$  s.t. $\delta_s - x_s = \delta_m  -  x_m$  and $\delta_m  -  x_m > 0$
        \item[$(\mathcal{R}_4)$]: $((s, m, x_s, x_m), \delta_m - x_m, (s, m', x_s + \delta_m - x_m, 0)) \in \Delta_\mathcal{D}$  if there exists $(m, \delta_m, m') \in \Delta_\mathcal{M}$  s.t. $ \delta_m  -  x_m < \delta_s - x_s$  and  $\delta_s \neq \infty$  and $\delta_m  -  x_m > 0$
        \item[$(\mathcal{R}_5)$]: $((s, m, x_s, x_m), \delta_m - x_m, (s, m', \infty, 0)) \in \Delta_\mathcal{D}$  if there exists $(m, \delta_m, m') \in \Delta_\mathcal{M}$  s.t.  $\delta_m  -  x_m < \delta_s - x_s$  and $\delta_s = \infty$  and $\delta_m  -  x_m > 0$
        \item[$(\mathcal{R}_6)$]: $((s, m, x_s, x_m), \delta_s - x_s, (s', m, 0, x_m + \delta_s - x_s)) \in \Delta_\mathcal{D}$ if there exists $(m, \delta_m, m') \in \Delta_\mathcal{M}$  s.t. $\delta_s - x_s < \delta_m  -  x_m$ and $\delta_m \neq \infty $ and $\delta_s - x_s > 0$
    \item[$(\mathcal{R}_7)$]: $((s, m, x_s, x_m), \delta_s - x_s, (s', m, 0, \infty)) \in \Delta_\mathcal{D}$  if there exists $ (m, \delta_m, m') \in \Delta_\mathcal{M}$ s.t. $\delta_s - x_s < \delta_m  -  x_m$  and $\delta_m = \infty$ and $\delta_s - x_s > 0$,\\
    where $\infty - x = \infty$ if $x$ is finite or infinite and $\infty + \infty  = \infty$.
    \end{itemize}
\item $(\nabla, x, \nabla) \in \lambda_\mathcal{D}$ for all $x \in I$   and $(\nabla, \infty, \nabla) \in \Delta_\mathcal{D}$
\end{itemize}
\end{definition}

\par The seven rules $\{\mathcal{R}_i\}_{i=1..7}$ introduce  input transitions and timeout transitions in $\mathcal{D}$. Each state $(s, m, x_s, x_m)$ of $\mathcal{D}$  is composed of a state $s$ of the specification, a state $m$ of the mutation machine, the value $x_s$ of the clock of  the specification  and the value $x_m$ of the clock of the mutation machine. The clock values are needed  for selecting timeout transitions. 
Input transitions are introduced with $\mathcal{R}_1$ and $\mathcal{R}_2$.
According to $\mathcal{R}_2$, there is a transition to accepting state $\nabla$ if different outputs are produced in $s$ and $m$ for the same input; otherwise the specification and the mutation machine move to next states, as described with rule $\mathcal{R}_1$.\\ 
Timeout transitions of the form $((s,m,x_s,x_m),\delta,(s',m', x_s',x_m'))$ are introduced by $\{\mathcal{R}_i\}_{i=3..7}$. The value of $\delta$ is the
delay  for reaching the only timeout $\delta_s$ defined in $s$ from  $(s,x_s)$ or a timeout defined in $m$ from $(m, x_m)$, since multiple timeouts can be defined in states of  nondeterministic mutation machines.   $\delta$ can be greater than the delays   for reaching  some timeouts defined in $m$; however, $\delta$  is never greater than $\delta_s - x_s$, the delay for reaching  the only timeout $\delta_s$  in $s$. 
So, $x_s'= 0$ if $\delta = \delta_s - x_s$; a similar statement holds for the clock and a selected timeout transition of the mutation machine.
In  $\mathcal{R}_3$, both the timeout in $s$ and a timeout in $m$ expire after $\delta$ time units.
In $\mathcal{R}_4$ only a timeout defined in  $m$ expires after $\delta$ time units and the only finite timeout defined in $s$  does not expire after  $\delta$ time units. A similar phenomenon is described with $\mathcal{R}_5$; but contrarily to $\mathcal{R}_4$,  the only timeout in $s$ is $\infty$ and  we set the clock of the specification  to $\infty$. Setting the clock to $\infty$ expresses the fact that we do not care any more about finite values of $x_s'$ because only the infinite timeout in $s$ must be reached. Without this latter abstraction on the values of $x_s'$, the size of $C$ could be infinite because we could have to apply $\mathcal{R}_4$   infinitely.   
The rules $\mathcal{R}_6$ and $\mathcal{R}_7$ are similar to $\mathcal{R}_4$ and $\mathcal{R}_5$, except that the only timeout  in $s$ expires before a timeout in $m$.
\par Each transition introduced in $\mathcal{D}$ with $\mathcal{R}_1$, $\mathcal{R}_2$, and $\mathcal{R}_3$, is defined by a transition of the specification and a transition of the mutation machine. A transition introduced with  $\mathcal{R}_4$ is defined by a timeout transition of the mutation machine  and the only timeout transition of the specification in $s$.
Every comb of $\mathcal{D}$ is defined by   a  comb of the specification and a comb of the mutation machine, i.e., it has been obtained by composing the transitions in a comb of the specification with the transitions in a comb of the mutation machine. 

\par An execution of $\mathcal{D}$ from a timed state $(c,x)$ is a sequence of steps between  timed states of $\mathcal{D}$; it can be defined similarly to that for a TFSM-T. An execution starting from $(c_0,0)$ and ending at $\nabla$ is called {\em accepted}. As for TFSMs-T, we can associate every execution of $\mathcal{D}$ with a comb and a timed input sequence. A comb of $\mathcal{D}$ is accepted if it corresponds to an accepted execution.

\begin{lemma}
A  comb $\pi$ of $\mathcal{M}$ is revealing if  it defines an accepted comb of $\mathcal{D}$. 
\end{lemma}
\begin{proof}
Let $\pi'$ be an accepted comb of $\mathcal{D}$ defined by a deterministic comb $\pi$ of  $\mathcal{M}$.  $\pi'$ includes the sink state and it was obtained by composing the transitions in $\pi$ with transitions in the specification with the rules in Definition~\ref{def:da}. $\pi$ corresponds to an execution of the mutation machine producing an unexpected timed output sequence, because otherwise $\pi'$ would not have been accepted. Consequently $\pi$ is revealing.
\end{proof}

A revealing comb  can be common to many mutants, in which case those mutants are said to be {\em involved} in the comb.  Mutants involved in a revealing comb are detected by the tests enabling the comb. We let $Susp_{X}$ denote the set of suspicious transitions  in  $X$.
\begin{lemma}\label{lemma:involved}
A mutant $\mathcal{P}$ is involved in a revealing comb $\pi$ of $\mathcal{M}$ if and only if  $Susp_\pi\subseteq Susp_\mathcal{P}$. 
 \end{lemma}
 \begin{proof}
 Any mutant which does not contain a transition in a comb is not able to perform any execution defining the comb. Mutants contain all trusted transitions but a selected subset of suspicious transitions specified in the mutation machine. Consequently a mutant is involved in a revealing comb  if and only if  the suspicious transitions in the comb are included in the suspicious transitions contained in the mutant.
 \end{proof}
 
 Lemma~\ref{lemma:involved} assumes that mutants are known and indicates how to check if a mutant is involved in a given (deterministic or nondeterministic) revealing comb. However, we want to avoid the enumeration of the mutants in eliminating the nonconforming mutants detected by a test; we also want to generate tests corresponding to deterministic revealing combs because they detect nonconforming mutants as stated in Lemma~\ref{lemma:rev:detected}. So we will focus on extracting only deterministic revealing combs of $\mathcal{M}$ from $\mathcal{D}$. We can obtain  deterministic revealing combs of $\mathcal{M}$
by performing a Breadth-first search of  sink state $\nabla$ in $\mathcal{D}$ while passing through transitions of $\mathcal{D}$ defined with transitions  of $\mathcal{M}$ which cannot be defined in an identical mutant.
 To compute $Rev_{\alpha}(\mathcal{M})$ the set of deterministic $\alpha$-revealing combs, we also apply Breath-first search of the sink state in the distinguishing automaton $\mathcal{D}$. This time the search is step-wise and guided by timed inputs in $\alpha$; 
it consists to pass a timeout transition in $\mathcal{D}$ whenever the delay between the current and the previous input in $\alpha$ is greater than the timeout of the transition or to pass an input transition in $\mathcal{D}$ when the current state in $\mathcal{D}$ defines a timeout smaller than delay between the current and the previous input. The result of the Breath-first search is $Rev_{\alpha}(\mathcal{M})$.
\begin{example}\label{example:revcomb}   Figure~\ref{fig:da} presents an excerpt of the distinguishing automaton with timeouts for the $\mathcal{S}_1$ and $\mathcal{M}_1$ in Figure~\ref{fig:tfsm}. It is the relevant  fragment for extracting the revealing combs for the test $\alpha = (b, 0.5)(a, 1)(b, 6.7)(a, 7.2)$. The whole distinguishing automaton is too big to fit in this paper.  $[t] (\mathcal{R})$ indicates an input/output transition or a timeout transition $t$ of the mutation machine defining  the transition of the automaton introduced with rule $\mathcal{R}$; e.g., the timeout transition  $((s_3, s_3, 0, 0), 5, (s_2, s_3, 0, 5))$ is defined by $t_{17}$ and the timeout of $t_{17}$ has not expired when the rule $\mathcal{R}_6$ is applied. There are  six deterministic $\alpha$-revealing combs: 
     ${\boldsymbol t_3} t_2 t_6 t_5 {\boldsymbol t_{17}} {\boldsymbol t_8} t_{12} {\boldsymbol t_{10}}$,
     ${\boldsymbol t_3} t_2 t_6 t_5 {\boldsymbol t_{17}} {\boldsymbol t_8} t_{12}{\boldsymbol t_{13}}$,
     ${\boldsymbol t_3} t_2 t_6 t_5 {\boldsymbol t_{17}} {\boldsymbol t_{15}} {\boldsymbol t_{17}} {\boldsymbol t_{14}}$,
     ${\boldsymbol t_{16}} t_2 t_6 t_5 {\boldsymbol t_{17}} {\boldsymbol t_8} t_{12}{\boldsymbol t_{10}}$,
     ${\boldsymbol t_{16}} t_2 t_6 t_5 {\boldsymbol t_{17}} {\boldsymbol t_8} t_{12}{\boldsymbol t_{13}}$ and
     ${\boldsymbol t_{16}} t_2 t_6 t_5 {\boldsymbol t_{17}} {\boldsymbol t_{15}} {\boldsymbol t_{17}}{\boldsymbol t_{14}}$, where 
 the  transitions  in bold are suspicious.
 \end{example}
 
\subsection{Encoding submachines involved in revealing combs}
We introduce a Boolean variable for each suspicious transition in 
mutation machine $\mathcal{M}$; Based on Lemma~\ref{lemma:involved}, we build Boolean formulas over 
these variables to encode the mutants involved in revealing comb. 
A solution of such a formula assigns a truth value to every transition variable. We say that a solution of a formula {\em determine} a submachine $\mathcal{P}$ of $\mathcal{M}$ if $\mathcal{P}$ is composed of the trusted  transitions and the suspicious transitions of which the values of the corresponding transition variable is $True$ in the solution.
In general, the submachine for the solution of a formula can be noninitially-
connected, nondeterministic or incomplete. Later we encode  mutants (deterministic and complete submachines) with  additional 
formulas. For now, 
let us encode the submachines involved in 
revealing combs of $\mathcal{M}$ with Boolean formulas.

\par Let $\alpha$ be a test and $Rev_\alpha(\mathcal{M}) = \{\pi_1,\pi_2,\ldots,\pi_n\}$ be the set of deterministic revealing combs of $\mathcal{M}$ enabled by  $\alpha$. We encode  a comb $\pi=t_1t_2\ldots t_m$ of $\mathcal{M}$ with the Boolean formula $ \varphi_\pi =\bigwedge_{t_i \in  Susp_{\pi}} t_i$, the conjunction of all the suspicious transitions in $\pi$. Clearly, any solution of $ \varphi_\pi$ determines a submachine (of the mutation machine) containing the comb $\pi$; The executions associated with $\pi$ are defined in such a submachine which is detected by $\alpha$. Conversely, each submachine determined by a solution of the negation of  $\varphi_\pi$ does not contains $\pi$; it cannot define any execution associated with $\pi$ and is not detected by $\alpha$.  Such a submachine is not necessarily a mutant because it can be nondeterministic or incomplete. 
For the set of deterministic revealing combs in $Rev_\alpha(\mathcal{M})$, let us define the formula $ \varphi_\alpha = \bigvee_{\pi \in Rev_\alpha(\mathcal{M})} \varphi_\pi$. The set of (possibly nondeterministic or incomplete) submachines of $\mathcal{M}$  detected by $\alpha$ is determined by a solution of $\varphi_\alpha$.  A submachine of $\mathcal{M}$ surviving $\alpha$ cannot contain any comb in  $Rev_\alpha(\mathcal{M})$ and it cannot be determined by a solution of the  negation of $\varphi_\alpha$, as stated in Lemma~\ref{lemma:surv:submachine}.
\begin{lemma} \label{lemma:surv:submachine}
A submachine of $\mathcal{M}$ survives a test $\alpha$  if and only if
it can be determined by a solution of $\neg \varphi_\alpha$.
\end{lemma}

To obtain the mutants surviving $\alpha$, we remove from the solutions of $\neg \varphi_\alpha$ those determining nondeterministic or incomplete submachines. This is possible with  a Boolean formula encoding only the mutants in $\mathcal{M}$.

\subsection{Encoding the mutants in a mutation machine}

Let  $T= {t_1, t_2,\ldots, t_n}$ be a set of Boolean variables for all the transitions $t_i$ of $\mathcal{M}$, $i=1..n$. Let us define the Boolean formula $\xi_T$ as follows
\[
\xi_T = \bigwedge_{k=1...n-1}(\neg t_k \vee \bigwedge_{l=k+1..n} \neg t_l) \wedge \bigvee_{k=1..n} t_i 
\]
A solution of $\xi_T$ assigns $True$ to exactly one selected variable and assigns $False$ to all other variables.
Note that $\xi_T$ is a CNF-SAT~\cite{DPLL62} formula and it can be solved using an existing SAT solver~\cite{Soos2009}.

\par Let $\mathcal{M}$ be a mutation machine for the specification machine $\mathcal{S}$. Clearly, $\lambda_\mathcal{S} \subseteq \lambda_\mathcal{M}$ and $\Delta_\mathcal{S} \subseteq \Delta_\mathcal{M}$. 
A deterministic and complete submachine of $\mathcal{M}$  selects one transition in  $\lambda_\mathcal{M}(s,i)$ and one transition in $\Delta_\mathcal{M}(s)$ for every state $s$ and input $i$; it is therefore determined by a solution of $\varphi_\mathcal{M}$ defined as follows.
\[
\varphi_\mathcal{M} = \bigwedge_{(s,i) \in S\times I} \xi_{\lambda_\mathcal{M}(s,i)} \wedge \bigwedge_{s \in S} \xi_{\Delta_\mathcal{M}(s)} \wedge \bigvee_{t\in \lambda_\mathcal{S} \cup \Delta_\mathcal{S}} \neg t
\]

\par The specification cannot be determined by a solution of $\varphi_\mathcal{M}$ because  its subformula $ \bigvee_{t\in \lambda_\mathcal{S} \cup \Delta_\mathcal{S}} \neg t$  encodes the rejection of transitions of the specification. The graph composed of the transitions selected by a solution can be disconnected, in which case it does not represent any mutant; a mutant can be obtained by extracting all the selected transitions connected to the initial state. We can prove Lemma~\ref{lemma:faultdomain} based on the previous discussion.
\begin{lemma}\label{lemma:faultdomain}
A submachine of  $\mathcal{M}$ is complete and deterministic if and only if it is determined by a solution of $\varphi_\mathcal{M}$.
\end{lemma}
We can prove the  following theorem thanks to Lemma~\ref{lemma:faultdomain} and Lemma~\ref{lemma:surv:submachine}. 
\begin{theorem}\label{theo:surv:mutant}
A mutant survives the test $\alpha$ if it is determined 
by a solution of $ \neg \varphi_\alpha \wedge \varphi_\mathcal{M}$.
\end{theorem}
\begin{example}
Considering the revealing combs for  $\alpha =(b, 0)(a, 0)\linebreak(b, 5)(a, 5)$, we use the suspicious transitions in the six revealing combs in Example~\ref{example:revcomb} at Page~\pageref{example:revcomb} to compute 
$\neg \varphi_{\alpha} = (\neg t_3 \vee \neg t_{17} \vee \neg t_8 \vee \neg t_{10}) \wedge (\neg t_3 \vee \neg t_{17} \vee \neg t_8 \vee \neg t_{13}) \wedge (\neg t_3 \vee \neg t_{17} \vee \neg t_{15} \vee \neg t_{14}) \wedge (\neg t_{16} \vee \neg t_{17} \vee \neg t_8 \vee \neg t_{10}) \wedge (\neg t_{16} \vee \neg t_{17} \vee \neg t_8 \vee \neg t_{13}) \wedge (\neg t_{16} \vee \neg t_{17} \vee \neg t_{15} \vee \neg t_{14})$. The mutant composed with the transitions $t_1, t_2, t_4, t_6, t_5, t_7, t_9, t_{15}, t_{13}, t_{12}, t_{11}$ and $t_{16}$ is determined by a solution of  $\neg \varphi_{\alpha}$ and it survives $\alpha$. The submachine with the transitions $t_1, t_2,t_3, t_4$ and $t_6$ is determined by another solution of  $\neg \varphi_{\alpha}$; however, it is neither a mutant nor a solution of $\varphi_{\mathcal{M}_1}$ defined as follows: \\
$\varphi_{\mathcal{M}_1} = (t_1) \wedge (t_2)  \wedge (t_5) \wedge (t_4) \wedge (\neg t_7 \vee \neg t_{14}) \wedge (t_7 \vee t_{14}) \wedge (\neg t_8 \vee \neg t_{15}) \wedge (t_8 \vee t_{15}) \wedge (\neg t_{10} \vee \neg t_{13}) \wedge (t_{10} \vee t_{13}) \wedge (t_{11}) \wedge (\neg t_3 \vee \neg t_{16}) \wedge (t_3 \vee t_{16}) \wedge (t_6) \wedge (\neg t_9 \vee \neg t_{17}) \wedge (t_9 \vee t_{17}) \wedge (t_{12}) \wedge (\neg t_1 \vee \neg t_2 \vee \neg t_4 \vee \neg t_5 \vee \neg t_7 \vee \neg t_8 \vee \neg t_{11} \vee \neg t_{10} \vee \neg t_3 \vee \neg t_9 \vee \neg t_6 \vee \neg t_{12})$.
\end{example}
The mutants surviving a test $\alpha$ can be partitioned into conforming mutants and nonconforming mutants which can only be detected with a test different from $\alpha$. Nonconforming mutants can be used to generate additional tests and upgrade the constraints. The generated test suite is complete if the solutions of the constraints determine only conforming mutants. This is the intuition of the test verification and generation methods below. The methods avoid a one-by-one enumeration of the mutants because a single test  eliminates many of them.
\section{Verifying and Generating a Complete Test Suite\label{sec:analysis:generation}}
\begin{algorithm}[!t]
\caption{Verifying the completeness of a given test suite\label{algo:checkseqgen}}
\SetKwInOut{Input}{Input~}
\SetKwInOut{Output}{Output~}
\underline{\textbf{Procedure} Verify\_completeness}  $(\varphi_{\mathit{fd}}, E, \mathcal{D})$\;
\Input{$\varphi_{\mathit{fd}}$, a Boolean expression specifying a fault domain}
\Input{$E$, a (possibly empty) set of tests}
\Input{$\mathcal{D}$, the distinguishing automaton of $\mathcal{M}$ and $\mathcal{S}$}
\Output{$\alpha \neq \varepsilon$, a test detecting a nonconforming mutant surviving $E$; $\alpha = \varepsilon$, if $E$ is a complete test suite}
{\em initialization :} 
$\varphi_{E} := \bigwedge_{\alpha \in E}\neg \varphi_\alpha $ 
$\quad \varphi_{\mathit{fd}} := \varphi_{\mathit{fd}} \wedge \varphi_{E}$ 
$\quad \varphi_\mathcal{P} :=$ False $\quad \alpha := \varepsilon$\;

\Repeat{$\alpha \neq \varepsilon$ or $\mathcal{P} = null$}{
      $\varphi_{\mathit{fd}} := \varphi_{\mathit{fd}} \wedge \neg \varphi_\mathcal{P}$\;
      $\mathcal{P} :=$ Determine\_a\_submachine($\varphi_{\mathit{fd}}$)\;
	\If{$\mathcal{P} \neq null$} {
          Build $\mathcal{D}_\mathcal{P}$, the distinguishing automaton of $\mathcal{S}$ and $\mathcal{P}$\;
          \eIf{$\mathcal{D}_\mathcal{P}$ has no sink state} {
          	$\varphi_\mathcal{P} := \bigwedge_{t\in \lambda_\mathcal{P} \cup \Delta_\mathcal{P}} t$\;
          }{
          	Set $\alpha$ to the timed input sequence of an accepted comb of the distinguishing automaton $\mathcal{D}_\mathcal{P}$\;
          }
       }
  }
\Return $(\varphi_{\mathit{fd}}, \alpha$)\;
\end{algorithm}
Let $E = \{\alpha_1,\alpha_2,\ldots,\alpha_n\}$ be a test suite and  $\langle \mathcal{S},\simeq,\Mut(\mathcal{M})\rangle$ be a fault model.
Our method for verifying whether $E$ is complete  works in three steps. First we build the Boolean expression $\bigwedge_{\alpha \in E}\neg \varphi_\alpha \wedge \varphi_\mathcal{M}$ encoding the mutants surviving $E$; this is based on Theorem~\ref{theo:surv:mutant}. 
Secondly, we use a solver to determine a mutant surviving the Boolean expression. Thirdly, we decide that $E$ is a complete test suite  if there is no mutant surviving $E$ or all the mutants surviving  the tests in $E$ are conforming. 
Procedure {\em Verify\_completeness} in Algorithm~\ref{algo:checkseqgen} implements the method. It makes a call to {\it Determine\_a\_submachine} for obtaining a mutant in a fault domain specified  with $\varphi_{\mathit{fd}}$. {\it Determine\_a\_submachine} can use an efficient SAT-solver to solve $\varphi_{\mathit{fd}}$ and build  mutants from solutions.  
{\em Determine\_a\_submachine} returns $null$ when the $\varphi_{\mathit{fd}}$ is unsatisfiable, i.e., the fault domain is empty. {\em Verify\_completeness} always terminates; this is because the size of the fault domain and the number of revealing combs for a test are  finite, and the SAT problem is decidable. 
\par Procedure {\em Generate\_complete\_test\_suite} in Algorithm~\ref{algo:generation} implements the iterative generation of a complete test suite. In each iteration step,  a new test is generated to  detect a  surviving mutant returned by {\em Verify\_completeness} if the mutant is nonconforming; otherwise the mutant is discarded  from the set of surviving mutants.  {\em Generate\_complete\_test\_suite} always terminates because there are finitely many mutants in the fault domain,  {\em Verify\_completeness} always terminates and the number of  surviving mutants   is reduced at every iteration step. 

\begin{example}
The result of an execution of  {\em Verify\_completeness} with input $E_{init} = \{ (b, 0.5)(a, 1)(b, 6.7)(a, 7.2)\}$  is the nonempty test $(a,3)$, which indicates that $E_{init}$ is not complete. 
We can generate additional tests to be added to $E$ and obtain a complete test suite.
An execution of  {\em Generate\_complete\_test\_suite} with  $E_{init}$ produces five tests detecting all the 31 mutants in the fault domain. The tests are the following: 
$(b, 0.5)(a, 1)(b, 6.7)(a, 7.2)$, 
$(a, 3)$,  
$(a, 4)(a, 8)$, 
$(b, 0)(a, 0)(b, 0)(a, 0)$ and
$(b, 0)(a, 0)\linebreak(a, 0)$.  
The generated test suite  includes  identical untimed sequences applied after different delays, i.e., the delays are needed for the fault detection.
\end{example}
\begin{algorithm}[!t]
\caption{Generating a complete test suite $E$ from  $E_{\mathit{init}}$ \label{algo:generation}}
\SetKwInOut{Input}{Input~}
\SetKwInOut{Output}{Output~}
\underline{\textbf{Procedure} Generate\_complete\_test\_suite}  $(E_{\mathit{init}}, \langle \mathcal{S}, \simeq, Mut(\mathcal{M})\rangle)$\;
\Input{$E_{init}$, an initial (possibly empty) set of timed input sequences}
\Input{$\langle \mathcal{S}, \simeq, Mut(\mathcal{M})\rangle$, a fault model}
\Output{$E$, a complete test suite for $\langle \mathcal{S}, \simeq, Mut(\mathcal{M})\rangle$}
Compute $\varphi_\mathcal{M}$, the boolean formula encoding all the mutants in $Mut(\mathcal{M})$\;
Build $\mathcal{D}$, the distinguishing automaton of $\mathcal{S}$ and $\mathcal{M}$\;
$\varphi_{\mathit{fd}} := \varphi_\mathcal{M}$\;
$E := \emptyset$\;
$E_{curr} := E_{init}$\;
\Repeat{$\alpha = \epsilon$}{
	$E := E \cup E_{\mathit{curr}}$\;
    $(\varphi_{\mathit{fd}}, \alpha) := $ Verify\_completeness$(\varphi_{\mathit{fd}}, E_{\mathit{curr}}, \mathcal{D})$\;
    $E_{\mathit{curr}} := \{ \alpha \}$\;
  }
\Return $E$\;
\end{algorithm}
\section{Experimental results\label{sec:exp}}
We implemented in the C++ language a prototype tool for an empirical evaluation
of the efficiency of the proposed methods. The experiment was realised with  a computer equipped with  the processor Intel(R) Core(TM) i5-7500 CPU @ 3.40 GHz and 32 GB RAM. The tool uses the solver 
cryptoSAT~\cite{Soos2009}. 
We present the results of the evaluation of the proposed methods with 
 randomly generated specifications and
 a specification of the trivial file transfer protocol. 
 
\subsection{Case of the trivial file transfer protocol}
We  consider a TFSM-T specification of the Trivial File Transfer Protocol (TFTP) \cite{RFC1350}.  TFTP is timeouts-dependent and it has already been tested in~\cite{Zhigulin2012}. Figure~\ref{fig:tftp} shows our  TFSM-T model for TFTP.  The model was designed according to the specification in ~\cite{RFC1350} and the modeling purposes  in~\cite{Zhigulin2012}. The modeling purposes 
 focus on the behavior of reading files. No more than three packages are transferred and the timeout for waiting for a packet equals three seconds. Moreover, we assume the file exists. Unlike the model in~\cite{Zhigulin2012} our model in Figure~\ref{fig:tftp} is complete and deterministic. For the sake of clarity, multiple transitions from one state to another  are represented with a single arrow.

\begin{figure}[t]
\centering
\scalebox{0.7}{
\begin{tikzpicture}[->,>=stealth',shorten >=1pt,auto,node distance=4cm, semithick]
    \node[state] (1)    at (6,5)                {Init};
    \node[state]         (2) at (1,3)  {Wait1};
    \node[state]         (3) at (6,1) {Wait2};
    \node[state]         (4) at (10,3) {Wait3};

    \path (1) edge node[midway, sloped, above] {RRQ/DATA1}         (2)
              edge [loop above] node[midway, right, xshift=0.2cm, yshift=0.3cm, align=left] {ACK1/Not defined \\ ACK2/Not defined \\ ACK3/Not defined \\ Error/Not defined \\ $\infty$}         (1)
          (2) edge [loop above] node[midway] {RRQ/Not defined}         (2)
          	  edge  node[midway, below left] {ACK1/DATA2}         (3)
              edge 	[bend left]node[midway, above left, align=left, xshift=0.7cm] {ACK2/ERROR \\ ACK3/ERROR \\ ERROR/Empty \\ $3$}         (1)
          (3) edge [loop below] node[midway, right, align=left] {RRQ/Not defined \\ ACK1/Ignore}         (3)
              edge              node[pos=0.3, below right] {ACK2/DATA3}   (4)
              edge 	node[pos=0.4, xshift=0.3cm, right,align=left] {ACK3/ERROR \\ ERROR/Empty \\ $3$}         (1)
          (4) 
              edge [loop below] node[midway, right, align=left] {RRQ/Not defined \\ ACK1/Ignore \\ ACK2/Ignore}         (4)
              edge[bend right] 	node[pos=0.3, above right, align=left] {ACK3/Empty \\ ERROR/Empty \\ $3$}         (1);
  \end{tikzpicture}
}
\caption{A TFSM-T modeling the TFTP; Init is the initial state; multiple transitions from one state to another  are represented with a single arrow.\label{fig:tftp}}
\end{figure}
\par  The tool generated within $0.31$s  a complete test suite of size $23$ for a mutation machine defining 1404928 mutants of  the TFSM-T in Figure~\ref{fig:tftp}. The maximal length of the tests is $5$. The mutation machine has $198$ mutated transitions. The mutated transitions were introduced as follows. Firstly,  we added $2$  finite timeouts and one mutated infinite timeout in all but the initial state; they are  $1$ and $5$. For each state and each input, we added a  transition to every state and for each output,  if the specified  output for the input is not "Not defined".
We noticed that the size of the generated complete test suite could  be reduced to 16 by removing the  seven tests which are  prefixes  of the others. Test suite optimization is a challenge to be addressed in further work.
\par We also generated a  complete chaos machine (the maximal  timeout is $5$) for the specification in Figure~\ref{fig:tftp}. 
The chaos machine defines an input/output transition from any state to any other  for each pair of input-output; it also defines a timeout transition from any state to any other  for each timeout between $1$ and $5$, and $\infty$. 
The resulting chaos mutation machine has $626$ more transitions than the specification, and have approximately $2.9 \times 10^{34}$ mutants. 
We generated 50 tests  within $1970$s; they can be reduced to 32 tests  by removing the tests' prefixes. The maximal length of the  tests is  $5$.  
\par We relaxed the modeling purpose by allowing $15$ packets instead of $3$. The corresponding  TFSM-T specification has $16$ states and we built a mutation machine with $9438$ mutated transitions defining $1.9 \times 10^{46}$ mutants. The tool generated a complete test suite of size $98$ within $555.14$s. The test suite can be amputated from $23$  tests' prefixes. The maximal length of the tests is $17$. 
\par  We have generated  complete test suites for fault domains of important sizes. \cite{ZhigulinYMC11} generates a complete test suite for a  chaos TFSM-T mutation machine defining less than $10^8$ mutants. The efficiency of our approach depends on the complexity of the mutation machine. In the average, the higher the size of the fault domain, the longer is the time for generating complete test suites.

\subsection{Case of randomly generated TFSMs-T}
\begin{table}[!t]
\centering
\begin{tabular}{|l|c|c|c|c|} \hline
 & \multicolumn{4}{|c|}{ \#mutants in the fault domain} \\ \hline
    \#states      & $\simeq 10^4$& $\simeq 10^{8}$ & $\simeq 10^{12}$ &  $\simeq 10^{18}$ \\ \hline \hline
 4 states &  (9, 0.04) & (26, 9.17)  & (30, 319.19) &  N/A  \\ \hline
 8 states &  (9, 0.5) & (19, 0.65)  & (32, 5.31)  &  (68, 864.06)  \\ \hline
 10 states &  (9, 4.73) & (20, 21.44)  & (30, 682.97)  &  (58, 250.72)  \\ \hline
 12 states &  (8, 56.36) & (20, 1.32)  & (25, 66.1)  &  (47, 593.07)   \\ \hline
 15 states &  (5, 168.24) & (17, 227.56)  & (33, 418.72)  &  (58, 64.55)   \\ \hline
\end{tabular}
\caption{Size of the generated complete test suites and generating time ; for an entry $(x,y)$, $x$ is the size of the test suite and  $y$ is the generating time in seconds \label{tab:bench}} 
\end{table} 
We randomly generated specification machines for given numbers of states and 
mutation machines for the specification machines. The generated specification and mutation machines have 2 inputs and 2 outputs. The maximal timeout in the specification machines is $3$ and the one in the mutation machines is $5$.  The result of the evaluation is presented in Table~\ref{tab:bench}. We have measured  the generating time for  each test suite  and  the size of each test suite. An entry $(x,y)$ of Table~\ref{tab:bench}  indicates the size of the test suite $x$  and the corresponding generating time $y$ in seconds (s). 
\section{Conclusion\label{sec:conc}}
We lifted a constraint solving-based test generation approach to generate complete test suite for fault models for TFSMs-T. We defined the distinguishing automaton with  timeouts  which is  used to build SAT constraints, verify the completeness  of  test suites  and generate complete test  suites. We implemented a prototype tool for the proposed test verification and generation methods. The empirical evaluation of the methods indicates that they apply on industrial-size TFSMs-T specifying real systems.
\par  Further work is in progress to reduce the size of the test suites and lift the proposed methods to TFSMs expressing time constraints beyond the timeouts.


\end{document}